\documentclass[english]{article}
\usepackage[T1]{fontenc}
\usepackage[latin9]{inputenc}
\usepackage{geometry}
\usepackage{verbatim}
\usepackage{float}
\usepackage{amsmath}
\usepackage{amsthm}
\usepackage{amssymb}
\usepackage{graphicx}
\usepackage{wasysym}

\makeatletter
\theoremstyle{plain}
\newtheorem{thm}{\protect\theoremname}
\theoremstyle{plain}
\newtheorem{prop}[thm]{\protect\propositionname}

\newcommand{\nstaticallart}{4650}
\newcommand{\nstaticreg}{240615}
\newcommand{\nstaticregart}{4420}
\newcommand{\sharehet}{17.1}
\newcommand{\sharecluster}{61}
\newcommand{\shareiid}{21.6}
\newcommand{\sharehcone}{98.1}
\newcommand{\sharehione}{15.6}
\newcommand{\ncoef}{3280}
\newcommand{\nreg}{608}
\newcommand{\nart}{155}

\newcommand{\ncoefboot}{1371}

\makeatother

\usepackage{babel}
\providecommand{\propositionname}{Proposition}
\providecommand{\theoremname}{Theorem}

\begin{document}
\title{From Replications to Revelations: Heteroskedasticity-Robust Inference}
\author{Sebastian Kranz, Ulm University\thanks{Special thanks to Lars Vilhuber, Ben Greiner and all other data editors:
without your awesome work, studies like this would not be possible.
Also many thanks to Michal Kolesar, James MacKinnon, Enrique Pinzone
and Michael Vogt for great discussions.}}
\date{December 2024\\
(First version November 2024)}
\maketitle
\begin{abstract}
Analysing the Stata regression commands from \nstaticregart\ reproduction
packages of leading economic journals, we find that, among the 40,571
regressions specifying heteroskedasticity-robust standard errors,
\sharehcone\% adhere to Stata's default HC1 specification. We then
compare several heteroskedasticity-robust inference methods with a
large-scale Monte Carlo study based on regressions from \nart\ reproduction
packages. Our results show that t-tests based on HC1 or HC2 with default
degrees of freedom exhibit substantial over-rejection. Inference methods
with customized degrees of freedom, as proposed by Bell and McCaffrey
(2002), Hansen (2024), and a novel approach based on partial leverages,
perform best. Additionally, we provide deeper insights into the role
of leverages and partial leverages across different inference methods.

\end{abstract}

\section{Motivation and Basic Insights}

A considerable body of literature has proposed and recommended different
specifications for heteroskedasticity-robust inference. For instance,
based on Monte Carlo evidence, Long and Ervin (2000) strongly recommend
for sample sizes below 250 observations HC3 standard errors rather
than Stata's robust default: HC1.\footnote{Similar recommendations to adopt more robust inference methods for
smaller samples have been made e.g. by MacKinnon and White (1985),
who introduced HC1, HC2 and HC3, Chesher and Jewitt (1987), Chesher
and Austin (1991), and Cattaneo, Jansson, and Newey (2018).}

Is this recommendation widely followed in empirical practice? No,
the opposite holds true. We analyse the code of regression commands
found in the Stata scripts of \nstaticregart\ reproduction packages
from leading economic journals. 40,571 regression commands specify
heteroskedasticity-robust standard errors and \sharehcone\% stick
with HC1 standard errors. Appendix A provides more details.

Would inference results substantially differ if the recommendation
were more widely followed? We proceed with a subset of \nreg\  regressions
from {\small{}\nart\ } reproduction packages that have fewer than
1000 observations and can be reproduced using the current version
of the toolbox \textit{repbox}.\footnote{The repbox tool chain consists of a series of R packages that I am
developing to facilitate and semi-automate methodological meta-studies.
This paper constitutes my first application of that toolchain. The
toolchain automatically reproduces Stata code found in reproduction
packages, performing essential tasks such as automatic file path correction.
Additionally, it systematically stores information from regression
commands and the underlying data sets, enabling the replication and
modification of these regressions in both Stata and R. The general
endeavor is complex and remains a work in progress. Currently, repbox
does not yet robustly work for all reproduction packages that are
theoretically reproducible. Thorough development and documentation
will take a lot more time. Another part of repbox and area of ongoing
work, which is not yet used in this paper, is the automatic mapping
of regressions from reproduction packages to the regression tables
displayed in the corresponding articles. Thankfully, the Deutsche
Forschungsgemeinschaft (DFG) supports future work on repbox as part
of the larger SocEnRep project where it will also benefit from input
from colleagues from computer sciences and social sciences.} Appendix B provides details on the sample selection.

For Table 1 we explore significance tests for the null hypothesis
that a true regression coefficient $\beta_{k}$ is zero. We run each
significance test twice: once with HC1 standard errors and once with
HC3 standard errors. HC3 standard errors and p-values are always larger
than their HC1 counterparts. Table 1 only considers the subsample
of the tests for which the HC1 p-value is below 5\% and shows for
different intervals of that p-value the fraction of tests that are
no longer significant at a 5\% level if one uses HC3 standard errors
instead.

\begin{table}
\caption{Significance at 5\% level for tests based on HC1 and HC3 robust standard
errors}

\begin{centering}
\begin{tabular}{lcccc}
\hline
\hline

Interval of & \multicolumn{2}{c}{$n \leq 250$} & \multicolumn{2}{c}{$250 < n \leq 1000$} \\
HC1 p-value & No. tests & Share HC3 $p \leq 0.05$ & No. tests & Share HC3 $p \leq 0.05$ \\

\hline
\\
(0,0.01]
  & 398 & 96.5\% & 253 &99.2\% \\
(0.01,0.02]
  & 73 & 82.2\% & 30 &96.7\% \\
(0.02,0.03]
  & 56 & 58.9\% & 34 &70.6\% \\
(0.03,0.04]
  & 54 & 29.6\% & 17 &58.8\% \\
(0.04,0.05]
  & 37 & 8.1\% & 22 &9.1\% \\

\hline
\end{tabular}
\par\end{centering}
\medskip{}

\textit{\small{}Note: }{\small{}The table includes only those significance
tests from the \nreg\  original regressions from \nart\  reproduction
packages whose p-values under HC1 standard errors were below 5\%.
For each p-value interval, it shows the number of these tests and
the fraction that remain statistically significant at the 5\% level
when HC3 standard errors are used instead. The first two columns show
results for regression with sample size $n\le250$ and the remaining
two columns for regressions with $250<n\leq1000$.}{\small\par}
\end{table}

The results show that the switch from HC1 to HC3 has substantial effects
on statistical significance. For instance, fewer than 10\% of HC1-based
tests with p-values between 4\% and 5\% remain significant when using
HC3. This result holds even for regressions with sample sizes ranging
from 250 to 1,000 observations.

While these results are insightful, several important questions remain.
For instance, to what extent are HC1-based p-values excessively low,
and to what extent are HC3-based p-values overly conservative? Beyond
HC3 standard errors, alternative methods for improving robust inference
have been proposed. These include HC4 standard errors, introduced
by Cribari-Neto (2004); HC2 standard errors with alternative calculations
of degrees of freedom for the t-test, suggested by Bell and McCaffrey
(2002) and Imbens and Kolesar (2016); Hansen's (2024) modified jackknife
estimator, which also extend the degrees of freedom adjustments proposed
by Bell and McCaffrey (2002); and wild bootstrap inference, as proposed
e.g. by Wu (1986) and Roodman et al. (2019).

Do some of these methods perform better than others across typical
situations encountered in economic analyses? How heterogeneous is
a method's performance across different situations? Do applied researchers
encounter scenarios where these established approaches systematically
fail, and novel methods offer meaningful improvements?

To address these questions, we employ a large-scale Monte Carlo study.
Monte Carlo simulations are a common tool in research on robust inference,
but they typically examine only a small set of regression specifications,
usually not based on real world data sets.\footnote{A notable exception is Young (2022), who conducts a large-scale Monte
Carlo study based on instrumental variable regressions extracted from
30 reproduction packages of economic articles. Already, in his previous
article, Young (2019), hand-collected reproducible regressions from
a large set of 53 reproduction packages, but did not yet base the
Monte Carlo simulations on those regressions. } We perform Monte Carlo studies based on a large set of \nreg\ OLS
regressions originally conducted in \nart\ different reproduction
packages of published economic articles. Ideally, our approach offers
insights that are representative of the situations typically encountered
by empirical researchers. Additionally, this broad scope enables us
to investigate heterogeneity by examining how the performance of robust
inference methods varies across the different original regression
specifications.

In a nutshell, for each original regression of the form
\begin{equation}
y^{o}=X\beta^{o}+\varepsilon^{o},
\end{equation}
we specify a custom data generating process
\begin{equation}
y=X\beta+\varepsilon
\end{equation}
with the same $n\times K$ matrix of explanatory variables $X$ as
in the original regression.%
{} The true coefficients, $\beta$, are set to zero. In line with the
original researchers' assumptions, error terms are heteroskedastic.
Concretely, we assume $\varepsilon_{i}\sim N(0,\sigma_{i}^{2})$ for
each observation $i=1,\ldots,n$. The specification of the standard
deviations, $\sigma_{i}$, of the error terms, $\varepsilon_{i}$,
is a detailed process performed separately for each original regression.
Initially, multiple candidate FGLS specifications, constructed using
random forests, are estimated and calibrated. Subsequently, one candidate
is selected by comparing the moments of the original OLS residuals
with those of the OLS residuals obtained from Monte Carlo simulations
of the various candidates. Details are provided in Appendix C.

For each original regression, we draw $M=10,000$ Monte Carlo samples
and compute the p-values for a t-test of the null hypothesis $\beta_{k}=0$
for up to 25 coefficients, $\beta_{k}$, per regression. Coefficients
of fixed effects dummies are not tested. Each of the \ncoef\ tested
coefficients from the \nreg\ original regressions constitutes a distinct
\textit{test situation}, indexed by $s$.

For each test situation, we compare different test specifications
$\tau\in\{\text{IID, HC1, HC2, \ensuremath{\ldots}}\}$, which vary
by the type of standard error and the specification of the degrees
of freedom used in the t-distribution. Mathematical background on
the different specifications is provided in Section 2.

Let $p_{\tau,s}(m)$ denote the realized p-value for Monte Carlo sample
$m=1,\ldots,M$ in specification $\tau$ and test situation $s$.
Our analysis focuses on the 5\% significance level. The simulated
rejection rate is defined as the proportion of Monte Carlo samples
for which the p-value is below 0.05:
\begin{equation}
\pi_{\tau,s}^{0.05}=\frac{1}{M}\sum_{m=1}^{M}I(p_{\tau,s}(m)\leq0.05)
\end{equation}
where $I(\cdot)$ is the indicator function. Since the null hypothesis
is true in all test situations, p-values should be uniformly distributed
under a correctly specified t-test. Consequently, the ideal value
of the rejection rate $\pi_{\tau,s}^{0.05}$ is 0.05.

We measure deviations from this ideal value using the excess and lack
of the rejection rate, defined as:
\begin{equation}
\text{excess}_{\tau,s}=\max\left(\pi_{\tau,s}^{0.05}-0.05,0\right),
\end{equation}
\begin{equation}
\text{lack}_{\tau,s}=\max\left(0.05-\pi_{\tau,s}^{0.05},0\right).
\end{equation}

While an excessive rejection rate increases the risk of false discoveries,
lack in rejection rates can lead to under-powered significance tests.
Excess is generally regarded as more problematic than an equally high
lack. However, opinions may differ regarding acceptable levels of
excess and the degree of lack one is willing to tolerate for a given
reduction in excess.

\begin{figure}
\begin{centering}
\includegraphics[scale=0.9]{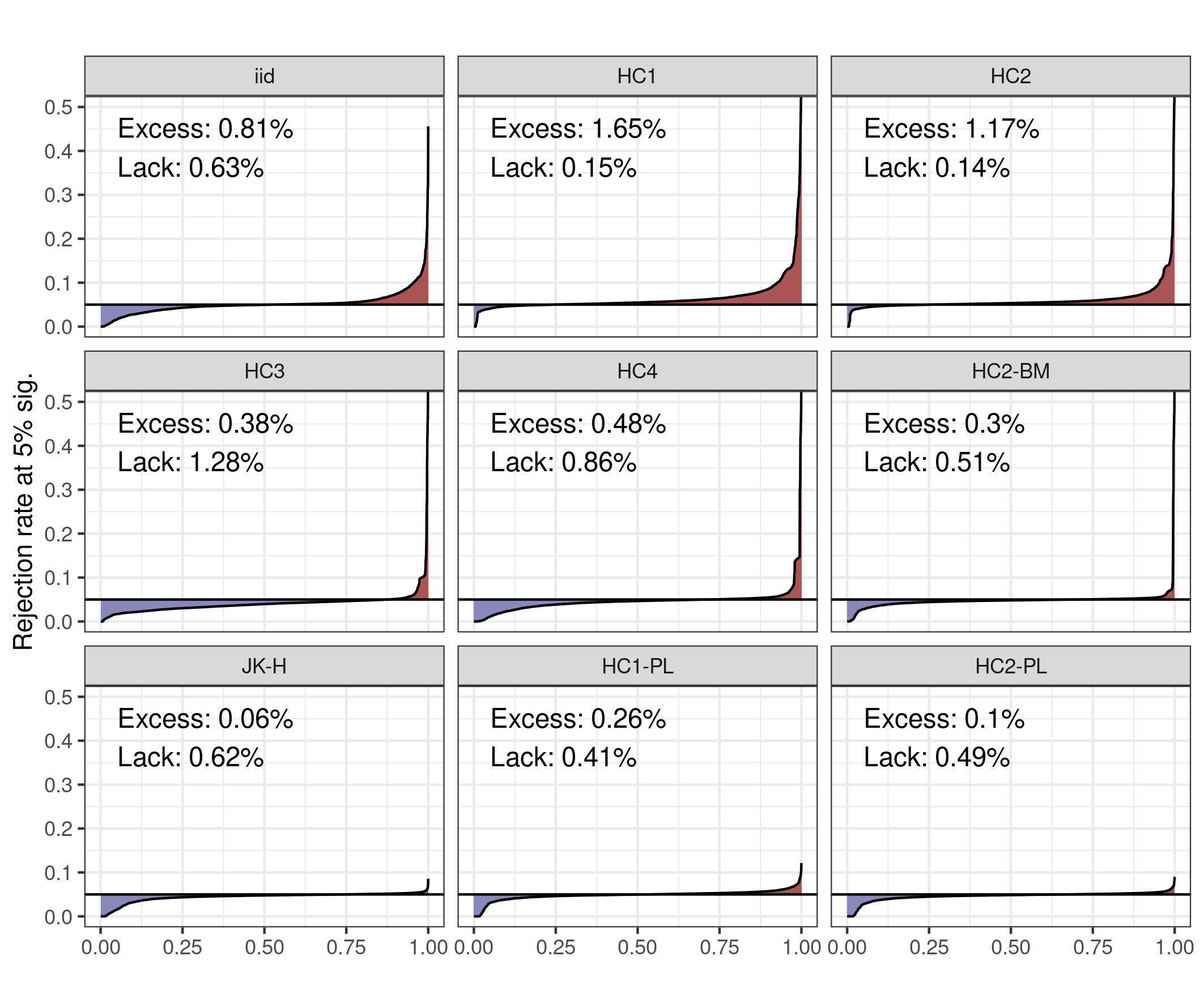}
\par\end{centering}
\caption{Core Results of Monte Carlo Study}
\medskip{}

\textit{\small{}Note: }{\small{}Each pane shows for a different specification
of standard errors and degrees of freedom the distribution of rejection
rate of t-tests with a 5\% significance level across \ncoef\ different
regression coefficients from \nreg\ regressions taken from \nart\ different
reproduction packages. Red areas correspond to regression coefficients
with excessive rejection rates (above 5\%) and blue areas to those
with lacking rejection rates (below 5\%). For each specification the
average excess and lack of the rejection rates across all regression
coefficients is reported.}{\small\par}
\end{figure}

Figure 1 shows for each specification $\tau$ the average excess and
lack and their distribution across all \ncoef\ test situations.
Consistent with conventional wisdom, average excess decreases when
moving in order from HC1, HC2, HC4, to HC3 standard errors, while
average lack correspondingly increases.

More surprisingly, in our sample of regressions with no more than
1000 observations, both HC1 and HC2 yield on average more excessive
rejection rates than inference based on i.i.d. standard errors, which
is consistent only under homoskedasticity.\footnote{Thus, adding the \textit{robust} option to a Stata \textit{regress}
command, such that HC1 standard errors are used, may, in smaller sample
sizes, misleadingly suggest that the resulting standard errors and
test results are more conservative than without the \textit{robust}
option.} While HC3 is most conservative in the sense of having the largest
average lack in rejection rates, it is not the specification with
the lowest average excess.

To partially order specifications, we say that one specification outperforms
another on average if it has a lower weighted sum of average excess
and average lack if the weight on excess is at least as large as the
weight on lack.

We study four specifications employing alternative degrees-of-freedom
adjustments. HC2-BM is based on the methods of Bell and McCaffrey
(2002) and Imbens and Kolesar (2016), JK-H refers to the jackknife
estimator proposed by Hansen (2024), and HC1-PL and HC2-PL implement
a novel partial leverage-based degrees-of-freedom adjustment introduced
in this paper. Each of these methods outperforms, on average, the
HC1 to HC4 specifications that rely on the default $n-K$ degrees
of freedom in the t-tests. This result highlights the importance of
correctly specifying degrees of freedom.

Among all methods considered, JK-H and HC2-PL outperform all others
on average, with no clear ranking between them. Both exhibit very
small average excess in their rejection rates of 0.06\% and 0.1\%,
respectively. For comparison, a simulation study assuming uniformly
distributed p-values yields an average excess in the rejection rates
of 0.09\%, which is attributable to noise in the Monte Carlo simulations.

Each pane of Figure~1 also plots the rejection rates $\pi_{\tau,s}^{0.05}$
of all \ncoef\ test situations, arranged in increasing order with
the quantile level on the x-axis. Blue and red shaded areas represent
rejection rates $\pi_{\tau,s}^{0.05}$ associated with positive lack
and excess, respectively. The average lack and excess for specification
$\tau$ correspond to the total size of the respective blue and red
areas.

There is substantial heterogeneity. No specification $\tau$ exhibits
uniformly excessive or uniformly lacking rejection rates across all
test situations. As shown on the left-hand side of each pane, all
specifications have rejection rates close to zero in some test situations.
Conversely, except for JK-H, HC1-PL and HC2-PL, all specifications
include a few test situations with extremely high rejection rates
exceeding 50\%.

\begin{figure}
\begin{centering}
\includegraphics[scale=0.9]{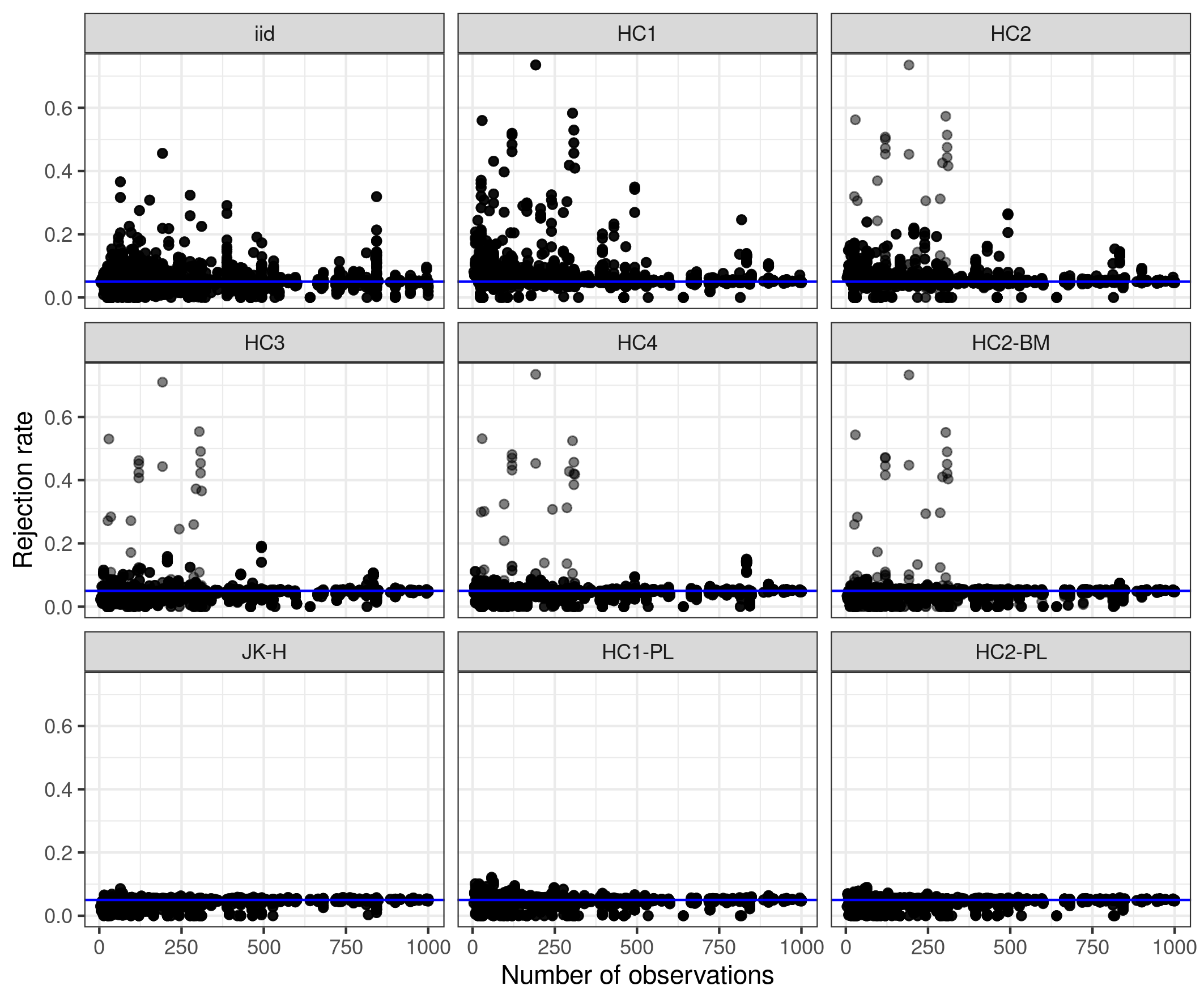}
\par\end{centering}
\caption{Rejection rates at 5\% level against sample size}

\textit{\small{}\medskip{}
}{\small\par}
\end{figure}

Figure 2 plots the rejection rates against the sample size of each
test situation. While highly excessive rejection rates are more likely
for smaller sample sizes, we find them also in test situations with
moderate sample sizes.

The next section provides mathematical background on robust inference
and deeper insights into the role of partial leverages and leverages.
We will also pin down the subset of test situations that have highly
excessive rejection rates for specifications HC2-HC4 and HC2-BM and
propose remedies. Section 3 then briefly concludes with additional
observations, like the extension of some insights to cluster robust
standard errors.

Due to the significant computational demands, wild bootstrap methods
are analyzed for only a subset of test situation, with the corresponding
results presented in Appendix D. While wild bootstrap specifications
outperform the conventional HC1 and HC2 specifications, they are outperformed
by the HC2-BM, JK-H, HC1-PL, and HC2-PL specifications that utilize
customized degrees of freedom.

\section{Mathematical Background and Deeper Insights}

This section provides theoretical background on heteroskedasticity-robust
inference and yields deeper insights into the role of partial leverages
and leverages. Consider the linear regression model:
\begin{equation}
y=X\beta+\varepsilon,\label{eq:reg_long}
\end{equation}
where $y$ is the dependent variable, $X$ is the matrix of explanatory
variables, $\beta$ is the vector of coefficients. The error terms
$\varepsilon_{i}$ are independently distributed with mean zero and
variance $\sigma_{i}^{2}$.

Using the Frisch-Waugh-Lovell (FWL) theorem, coefficient $\beta_{k}$
can be estimated with the simpler model:
\begin{equation}
\tilde{y}_{k}=\beta_{k}\tilde{x}_{k}+\tilde{\varepsilon}_{k},\label{eq:reg_fwl}
\end{equation}
where $\tilde{y}_{k}$ and $\tilde{x}_{k}$ are the residuals from
regressing $y$ and $x_{k}$, respectively, on all other explanatory
variables in $X$ except for $x_{k}$. Estimation of model (\ref{eq:reg_fwl})
yields the same OLS residuals $\hat{\varepsilon}$ and estimator $\hat{\beta}_{k}$
as the original model (\ref{eq:reg_long}).\footnote{The FWL representation also enhances computational efficiency of our
Monte Carlo simulations. While certain elements, in particular, the
leverages $h$ (introduced further below), cannot be derived from
the FWL representation, these elements need to be computed only once
per original regression. Since calculation of a full-blown $K\times K$
variance-covariance matrix can be avoided, the FWL representation
considerably accelerates computations that must be repeated for each
Monte Carlo sample.}

The partial leverage for observation $i$ with respect to explanatory
variable $x_{k}$ is defined as:
\begin{equation}
\tilde{h}_{k,i}=\frac{\tilde{x}_{k,i}^{2}}{\sum_{j=1}^{n}\tilde{x}_{k,j}^{2}},
\end{equation}
where $\tilde{x}_{k,i}$ is the $i$-th element of $\tilde{x}_{k}$.
Partial leverages satisfy $0\leq\tilde{h}_{k,i}\leq1$ and $\sum_{i=1}^{n}\tilde{h}_{k,i}=1$.\footnote{Furthermore, for $n\geq2$, we have $\tilde{h}_{k,i}<1$ if the original
regression model includes a constant.}

The FWL representation yields a simple formula for the true variance
of $\hat{\beta}_{k}$:
\begin{equation}
\text{Var}(\hat{\beta}_{k})=\frac{\sum_{i=1}^{n}\sigma_{i}^{2}\tilde{x}_{k,i}^{2}}{\left(\sum_{j=1}^{n}\tilde{x}_{k,j}^{2}\right)^{2}}=\frac{\sum_{i=1}^{n}\sigma_{i}^{2}\tilde{h}_{k,i}}{\sum_{j=1}^{n}\tilde{x}_{k,j}^{2}}.\label{eq:var_beta_k}
\end{equation}

The impact of error variance $\sigma_{i}^{2}$ on the variance of
$\hat{\beta}_{k}$ is thus proportional to the partial leverage $\tilde{h}_{k,i}$.

Heteroskedasticity-robust variance estimators for $\hat{\beta}_{k}$
of types HC0 to HC4 can all be expressed in the general form:
\begin{equation}
\hat{V}_{k}^{\tau}=\frac{\sum_{i=1}^{n}\left(\hat{\sigma}_{i}^{\tau}\right)^{2}\tilde{h}_{k,i}}{\sum_{j=1}^{n}\tilde{x}_{k,j}^{2}},\label{eq:V_beta_k}
\end{equation}
with
\begin{equation}
\hat{\sigma}_{i}^{\tau}=\sqrt{\alpha_{i}^{\tau}\hat{\varepsilon}_{i}^{2}},
\end{equation}
where $\hat{\varepsilon}_{i}$ is the OLS residual for observation
$i$, and $\alpha_{i}^{\tau}$ is an adjustment factor that depends
on the specification $\tau$.

Intuitively, the term $\hat{\sigma}_{i}^{\tau}$ estimates the unknown
standard deviation $\sigma_{i}$ of $\varepsilon_{i}$. Of course,
being based on a single residual, $\hat{\sigma}_{i}^{\tau}$ will
never consistently estimate $\sigma_{i}$. Yet, an important insight
of Eicker (1967), Huber (1967), and White (1980) is that, usually,
one can still consistently estimate the (co-)variances of $\hat{\beta}$.
Nowadays, their proposed estimator is often referred as HC0: it sets
$\alpha_{i}^{HC0}=1$. While consistent, HC0 can exhibit severe bias
in small samples and is rarely used in practice. MacKinnon and White
(1985) introduced three variants, HC1, HC2 and HC3, with improved
small-sample properties.

The HC1 adjustment factor is:
\begin{equation}
\alpha_{i}^{HC1}=\frac{n}{n-K}.
\end{equation}
This correction mirrors the degrees of freedom adjustment used in
the unbiased estimator $\hat{\sigma}^{2}=\frac{1}{n-K}\sum_{i=1}^{n}\hat{\varepsilon}_{i}^{2}$
for the error variance $\sigma^{2}=E[\varepsilon_{i}^{2}]$ under
homoskedasticity.

The HC2 adjustment factor is given by:
\begin{equation}
\alpha_{i}^{HC2}=\frac{1}{1-h_{i}},
\end{equation}
where $h_{i}$ is the \textit{leverage} of observation $i$, defined
as the $i$-th diagonal element of the hat matrix:
\begin{equation}
H=X(X^{\top}X)^{-1}X^{\top}.
\end{equation}

Leverages satisfy $0\leq h_{i}\leq1$ and $\sum_{i=1}^{n}h_{i}=K$.
HC2 is motivated by the property that, under homoskedasticity, $E[\hat{\varepsilon}_{i}^{2}]=(1-h_{i})\sigma^{2}$.
Furthermore, given certain regularity conditions Pustejovsky and Tipton
(2017) show that HC2 constitutes an unbiased variance estimator.

The HC3 adjustment factor is also based on leverages:
\begin{equation}
\alpha_{i}^{HC3}=\frac{1}{(1-h_{i})^{2}}.
\end{equation}
As shown in Hansen (2022), HC3-adjusted residuals $\sqrt{\alpha_{i}^{HC3}}\hat{\varepsilon}_{i}$
are equivalent to the leave-one-out prediction error $y_{i}-x_{i}\hat{\beta}^{(i)}$,
where $\hat{\beta}^{(i)}$ denotes the OLS estimator obtained from
the regression excluding observation $i$. The HC3 estimator can also
be interpreted as a jackknife estimator, satisfying\footnote{See MacKinnon et al. (2023) for a derivation. MacKinnon and White
(1985) introduced the jackknife formulation
\[
\hat{V}_{k}^{JK}=\frac{N-1}{N}\sum_{i=1}^{N}(\hat{\beta}_{k}^{(i)}-\bar{\beta}_{k}^{JK})^{2}
\]
as the HC3 variance estimator where $\bar{\beta}_{k}^{JK}$ denotes
the mean of the leave-one-out estimators $\hat{\beta}\halfnote_{k}^{(i)}$.
The current HC3 formulation was popularized by Davidson and MacKinnon
(1993).}
\begin{equation}
\hat{V}_{k}^{JK}=\sum_{i=1}^{N}(\hat{\beta}_{k}^{(i)}-\hat{\beta}_{k})^{2}.\label{eq:jk}
\end{equation}

Hansen (2024) adapts the jackknife estimator (\ref{eq:jk}) by employing
generalized Moore-Penrose inverses in cases an $\hat{\beta}^{(i)}$
cannot otherwise be computed. The HC4 estimator, introduced by Cribari-Neto
(2004), aims to better handle cases with high leverages by modifying
the adjustment factor to:
\begin{equation}
\alpha_{i}^{HC4}=\frac{1}{(1-h_{i})^{\delta}},
\end{equation}
where $\delta_{i}=\min\left\{ 4,\frac{nh_{i}}{K}\right\} $.

\subsection*{Partial leverage adjusted degrees of freedom}

In our specifications HC1 to HC4 we follow Stata and other software
packages and use $n-K$ degrees of freedom in the $t$-tests. Yet,
it is evident from (\ref{eq:V_beta_k}) that if partial leverages
are concentrated in a few observations, the standard errors are typically
imprecisely estimated as they depend strongly on the residuals of
those few observations. Choosing simply $n-K$ degrees of freedom
in the $t$-test does not account for such imprecisions arising from
highly concentrated partial leverages.

The idea behind our newly proposed methods HC1-PL and HC2-PL is to
base the degree of freedoms in the $t$-test on the concentration
of partial leverages. A widely used concentration measure in competition
policy is the Herfindahl-Hirschman index, defined as the sum of the
squared market shares of all competitors in a market. Similar to market
shares, partial leverages are non-negative and add up to 1. Thus the
concentration of partial leverages can be measured by the sum of squared
partial leverages. We denote the inverse of this measure as:

\begin{equation}
\tilde{n}_{k}=\left(\sum_{i=1}^{n}\tilde{h}_{k,i}^{2}\right)^{-1}
\end{equation}

We refer to $\tilde{n}_{k}$ as the partial-leverage-adjusted sample
size. It satisfies $1\leq\tilde{n}_{k}\leq n$. If all observations
have the same partial leverage then $\tilde{n}_{k}=n$, while $\tilde{n}_{k}=1$
in the limit case that the partial leverage becomes concentrated in
a single observation.

For HC1-PL and HC2-PL we specify the degrees of freedom in the $t$-test
as $\tilde{n}_{k}-1$ and use HC1 and HC2, respectively, as standard
errors. Appendix~E provides a justification for this partial-leverage-based
degrees of freedom by deriving it from a Satterthwaite approximation,
similar to the approach of Bell and McCaffrey (2002).

To understand why we specify the degrees of freedom as $\tilde{n}_{k}-1$
instead of $\tilde{n}_{k}$, consider the following. As long as the
original sample contains at least two observations, we have $\tilde{n}_{k}>1$.
Since a $t$-distribution can be defined for any fractional degrees
of freedom strictly greater than zero, this adjustment ensures properly
defined degrees of freedom whenever $n\geq2$. In the degenerate case
of a single observation, no standard error can be computed, and the
corresponding $t$-distribution with zero degrees of freedom becomes
degenerate. Using $\tilde{n}_{k}-1$ as the degrees of freedom allows
for continuous convergence to this degenerate case as $\tilde{n}_{k}\to1$.

Figure~\ref{fig:rejection_vs_pl} provides additional the empirical
support for our proposal. For the standard HC1-HC4 specifications,
as well as for specification HC2-BM: all test situations with highly
excessive rejection rates exhibit very small partial leverage adjusted
sample sizes. By imposing very wide t-distributions in those cases,
HC1-PL and HC2-PL effectively mitigate those excessive rejection rates.
Hansen's (2004) jackknife estimator also robustly handles cases with
small partial leverage.

\begin{figure}
\begin{centering}
\includegraphics[scale=0.9]{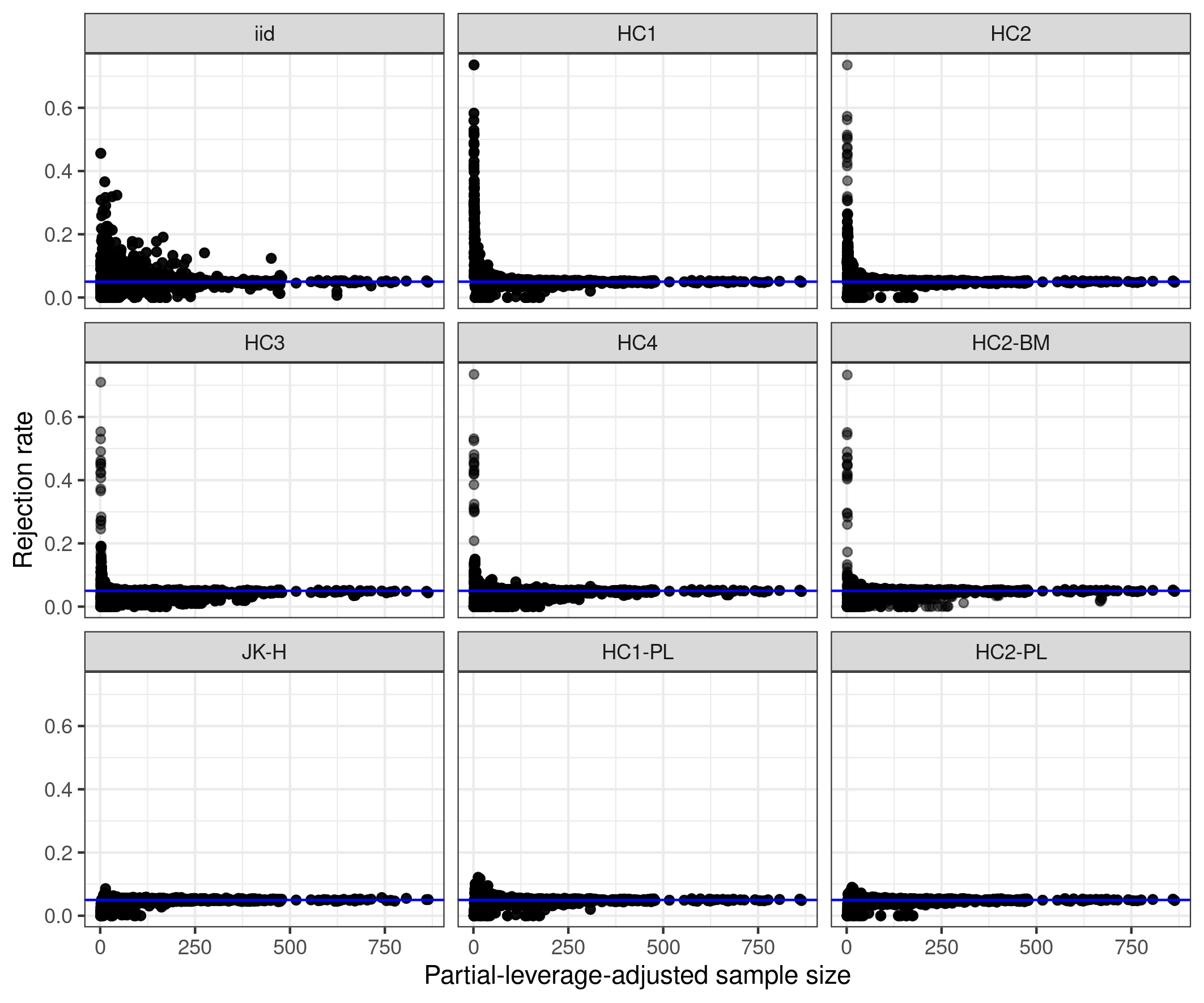}
\par\end{centering}
\caption{\label{fig:rejection_vs_pl}Rejection rates at 5\% level against partial-leverage-adjusted
sample sizes}

\textit{\small{}\medskip{}
}{\small\par}
\end{figure}

\subsection*{Full leverage and problematic regressions}

The insights discussed thus far do not yet provide a complete picture.
We now demonstrate that excessive rejection rates are particularly
pronounced in a subset of regressions involving observations with
a leverage $h_{i}$ of one and high partial leverages. We refer to
a leverage value of one as \textit{full leverage}. For cases with
full leverage, we propose a modification to the currently implemented
methods for computing standard errors.

Indeed, \sharehione\% of the regressions in our replication sample
include at least one observation with full leverage. Full leverage
occurs whenever one or several explanatory variables can perfectly
predict the corresponding observation $i$ such that $\hat{\varepsilon}_{i}=0$.
In that case $\hat{\varepsilon}_{i}$ provides no information about
$\sigma_{i}$. Our definitions above for the HC2, HC3, and HC4 variance
estimators are then not well-defined as $\hat{\sigma}_{i}^{\tau}$
involves the mathematically indeterminate expression $0/0$.

Fortunately, we don't need to predict $\sigma_{i}$ if $x_{k}$ is
not required to perfectly predict observation $i$, i.e. if a modified
regression that omits explanatory variable $x_{k}$ also has full
leverage for observation $i$. We then have $\tilde{x}_{k,i}=0$ and
also a zero partial leverage $\tilde{h}_{k,i}=0$. This means $\sigma_{i}$
does not affect the true variance of $\hat{\beta}_{k}$ and the general
variance estimator $\hat{V}_{k}^{\tau}$ specified in (\ref{eq:V_beta_k})
does not depend on $\hat{\sigma}_{i}^{\tau}$.

In contrast, if observation $i$ has full leverage and a positive
partial leverage $\tilde{h}_{k,i}>0$, explanatory variable $x_{k}$
is essential to perfectly predict observation $i$: $\sigma_{i}$
then influences the actual variance of $\hat{\beta}_{k}$ and variance
estimators $\hat{V}_{k}^{\tau}$ depend on which value is assumed
for $\hat{\sigma}_{i}^{\tau}$. The larger the partial leverage $\tilde{h}_{k,i}$,
the more relevant becomes the assumed value of $\hat{\sigma}_{i}^{\tau}$
for the variance estimator.

There is not yet a consistent treatment of cases with full leverage
across statistical software packages. The \textit{sandwich }package
in R (Zeileis, 2004) returns an \texttt{NaN} value for HC2 to HC4
standard errors if an observation has full leverage.

A similar non-invertibility problem arises in the context of cluster-robust
standard errors. Pustejovsky and Tipton (2017) theoretically justify
using a generalized Moore-Penrose inverse, which Kolesar (2023) and
the corresponding R package \textit{dfadjust} adapt to HC2 heteroskedasticity-robust
standard errors. The Moore-Penrose inverse of $1-h_{i}$ for $h_{i}=1$
is zero, i.e. one sets $\hat{\sigma}_{i}^{\tau}=0$ if $i$ has full
leverage. Poetscher and Preinerstorfer (2023) adopt this approach for
HC3 and HC4 as well.

Stata's approach is not yet well-documented. But experiments and personal
communications suggest that for observations with full leverage, Stata
evaluates $\alpha_{i}^{\tau}\hat{\varepsilon}_{i}^{2}=0/0$ numerically
for HC3, which yields results determined by rounding errors. For HC2
computation, as well as HC2-BM computation using Stata 18's new option
\textit{vce(hc2, dfadjust),} Stata seems to set $\hat{\sigma}_{i}^{\tau}=0$
for observations with full leverage.

Being in line with existing conventions, we have set $\hat{\sigma}_{i}^{\tau}=0$
for observations with full leverage in our Monte Carlo simulations
reported in Section 1. It is important to note that the theoretical
justification, that under this approach HC2 is then an unbiased estimator
of the true variance, only applies for the case that every observation
with full leverage has zero partial leverage. Yet, in those cases
also any other finite value for $\hat{\sigma}_{i}$ yields the same
variance estimator. For cases where the choice of $\hat{\sigma}_{i}$
matters, setting $\hat{\sigma}_{i}=0$ actually seems a recipe for
underestimating standard errors.

\begin{table}
\caption{\label{tab:lev}Average excess and lack for different subsamples and
methods}

\medskip{}

\begin{centering}
\begin{tabular}{lcccccccc}
\hline
\hline
& \multicolumn{2}{c}{No full} & \multicolumn{2}{c}{Full leverage} & \multicolumn{4}{c}{Full leverage, part. leverage > 0} \\
& \multicolumn{2}{c}{leverage} & \multicolumn{2}{c}{part. leverage = 0} & \multicolumn{2}{c}{$\hat \sigma_i = 0$} & \multicolumn{2}{c}{$\hat \sigma_i = \hat \sigma$} \\
& (1) & (2) & (3) & (4) & (5) & (6) & (7) & (8) \\
& excess & lack & excess & lack & excess & lack & excess & lack \\

\hline
HC1
  & 1.62\% & 0.09\%
  & 0.76\% & 0.33\%
  & 24.33\% & 1.03\%
  &  & \\
HC2
  & 1.05\% & 0.11\%
  & 0.58\% & 0.18\%
  & 23.86\% & 1.8\%
  & 3.13\% & 1.8\% \\
HC3
  & 0.22\% & 1.12\%
  & 0.06\% & 1.84\%
  & 20.53\% & 2.12\%
  & 2.73\% & 2.12\% \\
HC4
  & 0.3\% & 0.83\%
  & 0.17\% & 0.87\%
  & 22.13\% & 1.95\%
  & 2.92\% & 1.95\% \\
HC2-BM
  & 0.12\% & 0.49\%
  & 0.05\% & 0.53\%
  & 21.1\% & 2.3\%
  & 2.33\% & 2.3\% \\
JK-H
  & 0.06\% & 0.63\%
  & 0.06\% & 0.44\%
  & 0.02\% & 4.06\%
  &  &  \\
HC1-PL
  & 0.3\% & 0.31\%
  & 0.12\% & 0.61\%
  & 0.21\% & 4.37\%
  &  & \\
HC2-PL
  & 0.11\% & 0.44\%
  & 0.06\% & 0.48\%
  & 0.12\% & 4.56\%
  & 0\% & 4.56\% \\
\hline
Test & \\
situations &
\multicolumn{2}{c}{2544} &
\multicolumn{2}{c}{682} &
\multicolumn{2}{c}{31} &
\multicolumn{2}{c}{31} \\

\hline
\end{tabular}
\par\end{centering}
\medskip{}

\textit{\small{}Note:}{\small{} The table shows average excess and
average lack. Columns (1) and (2) show results for all test situations
where no observation has full leverage, columns (3) and (4) for test
situations where some observations $i$ have full leverage but for
all these observations the partial leverage is approximately zero.
Columns (5) to (8) show results for test situations that have observations
$i$ with full leverage and positive partial leverage. Columns (5)
and (6) set the corresponding $\hat{\sigma}_{i}$ to zero, while columns
(7) and (8) set the corresponding $\hat{\sigma}_{i}$ to $\hat{\sigma}$
the homoskedastic estimator of the standard deviation of $\varepsilon_{i}$
(only computed for standard errors based on HC2-HC4).}{\small\par}
\end{table}

Table \ref{tab:lev} shows the performance of different inference
specifications for different subset of test situations. Columns (1)
and (2) show results for tests without full leverage: HC2-BM and HC2-PL
perform both very well. Columns (3) and (4) study regressions that
have observations with full leverage, but all of them have zero partial
leverage. Those cases constitute 95.6\% of all test situations with
full leverage. Again both HC2-BM and HC2-PL perform both very well.
Columns (5) and (6) show results for the very small number of 31 test
situations where observations with full leverage are accompanied by
positive partial leverages. In those cases, all specifications except
for JK-H, HC1-PL and HC2-PL exhibit extremely excessive rejection
rates, with average excess well above 20\%.

We suggest selecting a different value for $\hat{\sigma}_{i}^{\tau}$
than 0 for cases with full leverage: one natural candidate is the
homoskedastic standard error estimate $\text{\ensuremath{\hat{\sigma}}}$.
As shown in Columns (7) and (8), this adjustment reduces average excess
of rejection rates to more reasonable levels. Of course, we cannot
rule out the possibility that, in some real-world data-generating
processes, observations with full leverage systematically exhibit
higher standard errors $\sigma_{i}$ than the homoskedastic estimate
$\hat{\sigma}$. Nonetheless, setting $\hat{\sigma}_{i}^{\tau}=\hat{\sigma}$
always yields more conservative results compared to setting $\hat{\sigma}_{i}^{\tau}=0$.

Interestingly, these 31 particular test situations are more diverse
than one might suspect. For example, only in 19 cases $x_{k}$ is
the obvious dummy variable that is equal to $1$ only for the observation
with full leverage and thus perfectly predicts it. If the regressions
had no control variables that are correlated with $x_{k}$ such a
dummy variable would exhibit a very high partial leverage for the
observation $i$ that is perfectly predicted by $x_{k}$. Yet, with
control variables, the structure of the FWL residual $\tilde{x}_{k}$
may substantially differ from $x_{k}$ and partial leverages may be
far less concentrated. In our sample one such dummy variable even
has a partial leverage of only 0.6\% at the perfectly predicted observation
and and a large partial-leverage adjusted sample size of $\tilde{n}_{k}=93.5$.
For these test situations, empirical rejection rates remain close
to 5\% for both HC2-PL and HC2-BM, regardless of whether $\hat{\sigma}_{i}^{\tau}$
is set to $0$ or $\hat{\sigma}$.

However, most of the 31 test situations exhibit highly concentrated
partial leverages. Specifically, for 20 test situations, we observe
$\tilde{n}_{k}<2$. This explains the very conservative rejection
rates of HC1-PL and HC2-PL reported in columns (5) to (8). Those conservative
rejection rates reflect high uncertainty about the actual value of
the standard error. Technically, by fixing for observations with full
leverage $\hat{\sigma}_{i}^{\tau}$ to $0$ or $\hat{\sigma}$, the
actual variance of the standard error can be much smaller than if
$\hat{\sigma}_{i}^{\tau}$ is determined by the noisy OLS residual
$\hat{\varepsilon}_{i}$. It is nonetheless prudent to treat these
cases as if we are highly uncertain as we don't know how well the
fixed $\hat{\sigma}_{i}^{\tau}$ approximates the unknown true $\sigma_{i}$.
Also Hansen's jacknife estimator JK-H has very conservative rejection
rates in those test situations.

\section{Concluding Remarks}

We conclude briefly with our key recommendations and additional remarks.

Overall, we recommend using Hansen's (2024) jackknife estimator or
JK-H, or the novel HC2-PL specification for heteroskedasticity-robust
inference, with HC2-BM being in most cases an equally good alternative.
A limitation of the HC2-PL specification, unlike JK-H and HC2-BM,
is that it cannot be readily extended to hypothesis tests involving
multiple coefficients. Computing the jackknife estimator JK-H likely
requires a longer runtime compared to HC2-PL and HC2-BM; however,
for a single regression, this is typically not a concern on modern
computers.

We further recommend the following modification for all specifications
but JK-H. For observations $i$ that have full leverage, set $\hat{\sigma}_{i}^{\tau}$
to a more conservative value than 0, such as the homoskedastic error
estimate $\hat{\sigma}$. Additionally, software packages could report
the percentage of the standard error computed using the homoskedastic
error estimate by summing the partial leverages of all observations
$i$ with full leverage.

MacKinnon et al. (2023b) propose reporting various diagnostic statistics
to evaluate the robustness of inference, including the distribution
of partial leverages. As compact concentrations measures, one could
report the partial leverage-adjusted sample sizes, $\tilde{n}_{k}$,
potentially accompanied by a warning for very low values.

All methods can be easily adapted to compute confidence intervals.
Moreover, for methods that compute customized degrees of freedom $\nu_{k}^{\tau}$,
one could follow the suggestion of Imbens and Kolesar (2016), and
report adjusted standard errors by multiplying the original standard
errors with $q_{\nu_{k}^{\tau}}(0.975)/q_{n-K}(0.975)$ where $q_{\nu}$
is the quantile function of the t-distribution with $\nu$ degrees
of freedom.

In our Monte Carlo studies, HC1-PL was relatively close behind the
performance of HC2-PL and JK-H. The \textit{reghdfe} Stata command
(Correia, 2017) and the \textit{fixest} R package (Berge, 2018), which
are widely used for fixed effects regressions, provide heteroskedasticity-robust
standard errors based solely on HC1. This is likely because extending
the performance gains from fixed-effects absorption to the computation
of hat values, which are required for HC2-, HC3-, or HC4-based standard
errors, or to JK-H estimation, is non-trivial.\footnote{Personal communication and the inclusion of HC2 standard errors with
absorbed fixed effects in the \textit{areg} function of Stata 18 suggest
that StataCorp might have developed a yet unpublished, performant
method for this computation.} In contrast, absorbed fixed effects do not pose any challenges for
computing partial leverages and HC1 standard errors. Therefore, HC1-PL
might be a promising alternative specification for regressions with
absorbed fixed effects.

This paper does not study robust inference methods proposed by Cattaneo
et al. (2018) and by Poetscher \& Preinerstorfer (2023). I still need
to better understand those methods and also need to figure out whether
a fast implementation for our large scale Monte Carlo study is possible.
If not included in a future version of this paper, hopefully once
the \textit{repbox} toolbox is fully developed and well documented,
it will facilitate such studies by more skilled researchers.

Variants of JK-H and HC2-BM have also been developed for cluster-robust
inference. Also HC1-PL and HC2-PL can be readily extended to cluster-robust
inference. The partial leverage of cluster $g$ is defined as the
sum of the partial leverages of all observations within cluster $g$.
The inverse Herfindahl-Hirschman index is then computed by treating
each cluster as a single observation, yielding a partial-leverage-adjusted
number of clusters, $\tilde{G}_{k}$. A comprehensive Monte Carlo
assessment for cluster-robust inference is planned for a separate
paper.

\section*{Bibliography}
\begin{itemize}
\item Athey, S., Tibshirani, J., \& Wager, S. (2019). ``Generalized random
forests''. The Annals of Statistics, 47(2), 1148.
\item Bell, R. M., \& McCaffrey, D. F. (2002). ``Bias reduction in standard
errors for linear regression with multi-stage samples''. Survey Methodology,
28(2), 169-182.
\item Berge, L. (2018). ``Efficient estimation of maximum likelihood models
with multiple fixed-effects: the R package FENmlm''. CREA Discussion
Papers.
\item Cattaneo, M. D., M. Jansson, and W. K. Newey. 2018. ``Inference in
linear regression models with many covariates and heteroscedasticity''.
Journal of the American Statistical Association 113: 1350--1361.
\item Chesher, A., and I. Jewitt. 1987. ``The bias of a heteroskedasticity
consistent covariance matrix estimator''. Econometrica 55: 1217--1222.
\item Chesher, A., and G. Austin. 1991. ``The finite-sample distributions
of heteroskedasticity robust Wald statistics''. Journal of Econometrics
47: 153--173.
\item Correia, Sergio. 2017. ``Linear Models with High-Dimensional Fixed
Effects: An Efficient and Feasible Estimator''. Working Paper. http://scorreia.com/research/hdfe.pdf
\item Cribari-Neto, F. (2004). ``Asymptotic inference under heteroskedasticity
of unknown form''. Computational Statistics \& Data Analysis, 45(2),
215-233.
\item Christensen, Rune Haubo B (2018), ``Satterthwaite\textquoteright s
Method for Degrees of Freedom in Linear Mixed Models''. Notes for
the R package lmerTestR
\item Davidson, R., \& MacKinnon, J. G. (1993). ``Econometric Theory and
Methods''. Oxford University Press
\item Ding, P. (2021). ``The Frisch--Waugh--Lovell theorem for standard
errors''. Statistics \& Probability Letters, 168, 108945.
\item Eicker, Friedhelm (1967). ``Limit Theorems for Regression with Unequal
and Dependent Errors''. Proceedings of the Fifth Berkeley Symposium
on Mathematical Statistics and Probability. Vol. 5. pp. 59--82.
\item Hansen, B. (2022). ``Econometrics''. Princeton University Press.
\item Hansen, B. (2024). ``Jackknife standard errors for clustered regression.''
University of Wisconsin.
\item Huber, Peter J. (1967). \textquotedbl The behavior of maximum likelihood
estimates under nonstandard conditions\textquotedbl . Proceedings
of the Fifth Berkeley Symposium on Mathematical Statistics and Probability.
Vol. 5. pp. 221--233.
\item Imbens, G. W., \& Kolesar, M. (2016). ``Robust standard errors in
small samples: Some practical advice''. Review of Economics and Statistics,
98(4), 701-712.
\item Kolesar, M. (2023), ``Robust Standard Errors in Small Samples''.
Vignette of the R package dfadjust.
\item Long, J. S., \& Ervin, L. H. (2000). Using Heteroscedasticity Consistent
Standard Errors in the Linear Regression Model. The American Statistician,
54(3), 217--224.
\item Poetscher, B. M., \& Preinerstorfer, D. (2023). ``Valid Heteroskedasticity
Robust Testing''. Econometric Theory, 1--53.
\item Pustejovsky, J. E., \& Tipton, E. (2017). ``Small-Sample Methods
for Cluster-Robust Variance Estimation and Hypothesis Testing in Fixed
Effects Models''. Journal of Business \& Economic Statistics, 36(4),
672--683.
\item MacKinnon, James G.; White, Halbert (1985). ``Some Heteroskedastic-Consistent
Covariance Matrix Estimators with Improved Finite Sample Properties''.
Journal of Econometrics. 29 (3): 305--325.
\item MacKinnon, J. G., Nielsen, M., \& Webb, M. D. (2023a). ``Fast
and reliable jackknife and bootstrap methods for cluster-robust inference''.
Journal of Applied Econometrics, 38(5), 671-694.
\item MacKinnon, J. G., Nielsen, M., \& Webb, M. D. (2023b). ``Leverage,
influence, and the jackknife in clustered regression models: Reliable
inference using summclust''. The Stata Journal, 23(4), 942-982.
\item Roodman, D., Nielsen, M., MacKinnon, J. G., \& Webb, M. D. (2019).
``Fast and wild: Bootstrap inference in Stata using boottest''.
The Stata Journal, 19(1), 4-60.
\item Satterthwaite, F. E. (1946), \textquotedblleft An Approximate Distribution
of Estimates of Variance Components,\textquotedblright{} Biometrics
Bulletin 2, 110--114.
\item White, Halbert (1980). ``A Heteroskedasticity-Consistent Covariance
Matrix Estimator and a Direct Test for Heteroskedasticity''. Econometrica.
48 (4): 817--838.
\item Wu, C. F. J. (1986). ``Jackknife, bootstrap and other resampling
methods in regression analysis''. the Annals of Statistics, 14(4),
1261-1295.
\item Young, A. (2019). ``Channeling fisher: Randomization tests and the
statistical insignificance of seemingly significant experimental results''.
The Quarterly Journal of Economics, 134(2), 557-598.
\item Young, A. (2022). ``Consistency without inference: Instrumental variables
in practical application''. European Economic Review, 147, 104-112.
\item Zeileis, A. (2004). ``Econometric Computing with HC and HAC Covariance
Matrix Estimators''. Journal of Statistical Software, 11(10), 1--17.
\end{itemize}

\section*{Appendix A: Static code analysis of regression commands in Stata}

\setcounter{table}{0}
\renewcommand{\thetable}{A\arabic{table}}
\setcounter{figure}{0}
\renewcommand{\thefigure}{A\arabic{figure}}

We base our code analysis on \nstaticallart\ reproduction packages
containing Stata scripts from articles published in leading economic
journals. The sample is based on the reproduction packages downloaded
to generate my web app \textit{Finding Economic Articles with Data}
accessible under https://ejd.econ.mathematik.uni-ulm.de/. The coverage
for most journals is wide, but not perfect. Automated analysis of
the article pages to determine the download link of the reproduction
package and the subsequent download sometimes encountered problems.
Also, in order to reduce strain on computational resources and network
bandwidth, very large reproduction packages were often not downloaded.
The publication years of the corresponding articles range from 2005
to 2024. We utilize a custom-designed parser to systematically extract
information from all code lines in the Stata scripts, with particular
emphasis on lines containing regression commands.

Table A1 reports the frequency of the 10 most common regression commands
identified in our sample.

\begin{table}[H]
\centering{}\caption{Top 10 of most used Stata regression commands}
\medskip{}
\begin{tabular}{lcc}
\hline
Regression command & No. reproduction packages & No. code lines \\
\hline
regress & 3866 & 146847 \\
areg & 1416 & 41373 \\
xtreg & 859 & 19333 \\
reghdfe & 763 & 29921 \\
ivregress & 589 & 11975 \\
ivreg2 & 498 & 11016 \\
probit & 474 & 4913 \\
logit & 371 & 3810 \\
xtivreg2 & 213 & 5539 \\
dprobit & 165 & 2200 \\
\hline
\end{tabular}
\end{table}

We restrict our analysis to OLS regressions performed using one of
the following Stata commands: \textit{regress}, \textit{areg}, \textit{xtreg},
or \textit{reghdfe}. These commands appear in a total of \nstaticreg\ code
lines across \nstaticregart\ reproduction packages. Table A2 provides
a detailed breakdown of their distribution across journals.

\begin{table}[h]
\centering{}\caption{Numbers of reproduction packages and regression commands by journal}
\medskip{}
\begin{tabular}{lcc}
\hline
Journal & Reproduction packages & Regression commands \\
\hline
aer & 1085 & 57834 \\
aejapp & 552 & 37463 \\
aejpol & 503 & 31741 \\
restat & 497 & 25838 \\
pandp & 244 & 3026 \\
ms & 241 & 12600 \\
aejmac & 202 & 13052 \\
restud & 169 & 9822 \\
jpe & 152 & 9971 \\
jole & 143 & 9079 \\
jeea & 135 & 8069 \\
aejmic & 104 & 3004 \\
ecta & 92 & 5609 \\
jep & 86 & 2533 \\
aeri & 73 & 2749 \\
qje & 72 & 5305 \\
jaere & 70 & 2920 \\\hline
Total & 4420 & 240615 \\
\hline
\end{tabular}

\end{table}

The standard errors used in these regressions fall into three main
categories: \sharehet\% heteroskedasticity-robust, \sharecluster\%
cluster-robust, and \shareiid\% homoskedastic. Panel A of Table 1
reports the absolute numbers for each category.\footnote{For 791 command lines, the type of standard errors cannot be determined
through our static code analysis, as it depends on Stata macros that
are only resolved during runtime.}

\begin{table}[h]
\begin{centering}
\caption{Static Code Analysis: Frequency of Standard Errors}
\medskip{}
\par\end{centering}
\centering{}
\begin{tabular}{lcc}
\hline
\hline
& No. reproduction packages & No. code lines \\
\hline

\multicolumn{3}{l}{Panel A: Standard error category} \\
cluster & 3055 & 144882 \\
iid & 2937 & 51230 \\
robust & 1689 & 40571 \\
unknown & 47 & 791 \\
\\
\multicolumn{3}{l}{Panel B: Type of heteroskedasticity-robust standard error} \\
hc1 & 1660 & 39782 \\
bootstrap & 54 & 765 \\
hc3 & 3 & 24 \\

\hline
\end{tabular}
\end{table}

Panel B of Figure 1 illustrates the frequency of specific types of
heteroskedasticity-robust standard errors: \sharehcone\% of regression
commands use Stata's default robust standard error, HC1.

\section*{Appendix B: Selection of studied regressions and tested coefficients}

From the reproduction packages containing an OLS regression with robust
standard errors, we select the underlying regressions for the analysis
shown in Table 1 and the subsequent Monte Carlo studies based on the
following criteria:
\begin{itemize}
\item The size of the reproduction package's ZIP file is below 10 MB.
\item The regression command can be successfully run. The most common cause
of a run error is the absence of confidential or proprietary data
sets in the reproduction packages.
\item A successful run of an automatic translation of the regression to
R, based on the extracted information about the original regression.
The second run must yield the same results as the original Stata run,
except for small numerical discrepancies.
\item The entire reproduction, including the extraction and storage of regression
related information and the 2nd reproduction run in R for all regressions,
takes less than 15 minutes for the whole reproduction package. Reproduction
runs that take longer are currently canceled with a timeout.
\item From each reproduction package with regressions meeting these criteria,
at most four regressions are selected for the Monte Carlo simulations.
\item For each regression, $t$-tests are performed on at most 25 regression
coefficients. Each regression is estimated with a Least Squares Dummy
Variable (LSDV) specification, even if the original Stata command
uses absorbed fixed effects, to enable the computation of proper HC2,
HC3, and HC4 standard errors. $t$-tests are not performed for originally
absorbed fixed effects. Additionally, heuristics are used to identify
dummy variable sets that are not absorbed but look like fixed effects.
For these dummy variables no t-tests are performed either.
\end{itemize}
It is somewhat disappointing that we currently end up with only \nart\
reproduction packages containing at least one regression that satisfies
all the criteria above. While some issues, such as missing data in
a reproduction package, cannot be resolved, there remain additional
points of failure in the process of reproducing the original Stata
regression in R. Hopefully, some failures can be avoided in more robust
future versions of the repbox toolbox.

The number of test situations could be easily increased by selecting
more than the current maximum of four regressions from each reproduction
package. However, this approach risks skewing the results, as they
may become more strongly influenced by a smaller number of reproduction
packages that contain a large number of regressions.

\section*{Appendix C: Specifying DGPs for our Monte Carlo Studies}

Assume the model of the original regression $r$ estimated in the
reproduction package is given by
\begin{equation}
y^{r,o}=X\beta^{r,o}+\varepsilon^{r,o}.
\end{equation}
The corresponding Monte Carlo samples will be generated from the model:
\begin{equation}
y^{r}=X\beta^{r}+\varepsilon^{r}.
\end{equation}
The explanatory variables $X$ remain unchanged from the original
sample and we set $\beta^{r}=0$. To determine the distribution of
the error terms $\varepsilon^{r}$ for each Monte Carlo model, we
use the following general procedure:
\begin{enumerate}
\item For each original regression $r$, we specify a set of $C$ candidate
models $\mathcal{M}^{r,c}$ indexed by $c=1,...,C$ for the distribution
of the error term. We only consider candidate models with independently
distributed error terms for Monte Carlo sample $m$ satisfying $\varepsilon_{i}^{m,r,c}\sim N(0,\left(\sigma_{i}^{r,c}\right)^{2})$.
This means each candidate model is fully characterized by the specified
vector of standard errors $\sigma^{r,c}$ of the error term.
\item From each candidate model we draw $m=1,...,M$ samples of the error
term $\varepsilon^{m,r,c}$ and compute the corresponding OLS residuals
$\hat{\varepsilon}^{m,r,c}.$
\item We then compute for each candidate model a distance $d^{r,c}$ between
the original OLS residuals $\hat{\varepsilon}^{r,o}$ and the set
of Monte Carlo residuals $\{\hat{\varepsilon}^{m,r,c}\}_{m=1}^{M}$.
\item For each original regression $r$, we pick that candidate model for
the Monte Carlo DGP that has the lowest distance $d^{r,c}$.
\end{enumerate}
One candidate model assumes purely homoskedastic error terms, while
all other candidate models estimate error term standard errors $\sigma_{i}^{r,c}$
using a non-parametric FGLS specification based on random forests.
The dependent variable of the random forest is the absolute value
of the original OLS residuals, $\left|\hat{\varepsilon}_{i}^{r,o}\right|$.\footnote{Alternative specifications for the dependent variable would also be
reasonable, such as HC2- or HC3-adjusted residuals. Additionally,
observations with $h_{i}=1$, which have $\hat{\varepsilon}_{i}^{r,o}=0$,
could be excluded when training the random forests. However, due to
computational constraints, we have not systematically explored these
variations. Preliminary experiments suggest that these variations
have little impact on the main results.} The explanatory variables of the random forest are the same as the
explanatory variables in the original regression. Compared to a feasible
GLS specification based on a log-linear regression model, random forests
offer greater flexibility and are well-suited to capture non-linear
effects and interactions among explanatory variables in predicting
heteroskedasticity. Since random forest predictions are always computed
as averages of the dependent variable in the training data set, they
cannot produce negative standard errors.

Whether the random forest predicts larger or smaller degrees of heteroskedasticity
(measured by the variation in $\sigma_{i}$) depends on how the random
forest is trained and how predictions are performed. Ordered from
typically larger to smaller predicted heteroskedasticity, we consider
five different candidates: in-sample prediction, out-of-bag prediction,
out-of-bag prediction based on honest trees (see Athey et al., 2019),
an equally weighted linear combination of out-of-bag prediction for
honest trees and a homoskedastic model, and a purely homoskedastic
model.

No single approach provides the best fit for all original regression
specifications. One indicator of the approximation quality is the
similarity between the sample distribution of the original OLS residuals
and the distribution of the Monte Carlo OLS residuals. To evaluate
candidate models, we particularly focus on the kurtosis of the OLS
residuals: if the original error terms are i.i.d. and normally distributed,
the OLS residuals exhibit a kurtosis close to 3, whereas heteroskedasticity
generally increases the kurtosis.

More formally, let $\kappa^{m,r,c}$ denote the kurtosis of the OLS
residuals in Monte Carlo sample $m$ for candidate model $c$ of original
regression $r$. Let $\bar{\kappa}^{r,c}$ and $s_{\kappa}^{r,c}$
denote the corresponding mean and standard deviation of the kurtoses
across all $M$ Monte Carlo samples. A basic distance measure for
candidate model $c$ for original regression $r$ is the standardized
distance
\begin{equation}
\delta^{c,r}=\frac{\left|\bar{\kappa}^{r,c}-\kappa^{r,o}\right|}{s_{\kappa}^{r,c}}
\end{equation}
where $\kappa^{r,o}$ is the kurtosis of the original OLS residuals.
To some extent, this distance measure may favor data-generating processes
that produce OLS samples with high variability in kurtosis. To counteract
this effect, we also consider a modified distance measure:
\begin{equation}
\bar{\delta}^{c,r}=\frac{\left|\bar{\kappa}^{r,c}-\kappa^{r,o}\right|}{\bar{s}_{\kappa}^{r}}
\end{equation}
 where $\bar{s}_{\kappa}^{r}$ denotes the median of the kurtosis
standard deviations across all candidate models for the original regression
$r$. The final distance measure used is the average of these two
measures:
\begin{equation}
d^{c,r}=0.5\delta^{c,r}+0.5\bar{\delta}^{c,r}
\end{equation}
The chosen DGP for original regression $r$ is the candidate model
$c$ with the lowest distance $d^{c,r}$.

We do not base the selection of DGP on the standard deviations of
residuals. Instead, all candidate models are calibrated to produce
OLS residuals with standard deviations similar to those of the original
OLS residuals. This calibration is achieved by scaling the initially
obtained values of $\sigma^{r,c}$ by the ratio of the standard deviations
of the original residuals and the Monte Carlo residuals, $sd(\hat{\varepsilon}^{r,o})/sd(\hat{\varepsilon}^{m,r,c})$.

\section*{Appendix D: Wild Bootstrap Inference}

\setcounter{table}{0}
\renewcommand{\thetable}{D\arabic{table}}
\setcounter{figure}{0}
\renewcommand{\thefigure}{D\arabic{figure}}

Wild bootstrap techniques have gained a lot of attention in the context
of cluster robust inference, see e.g. Roodman et al. (2019) and MacKinnon
et al. (2023). For heteroskedasticity-robust inference they have been
proposed already by Wu (1986). We compute wild bootstrap p-values
for the t-test with null hypothesis $\beta_{k}=0$ as follows:
\begin{enumerate}
\item Estimate the restricted regression model under the null hypothesis
$\beta_{k}=0$ to obtain OLS residuals $e^{r}$ and predicted values
$\hat{y}^{r}$.\footnote{Notably, given the restriction $\beta_{k}=0$ the restricted OLS residuals
$e$ are simply the residuals of the OLS regression that leaves out
the regressor $x_{k}$. Comparing to the FWL representation in Section
2, we find $e^{r}=\tilde{y}_{k}$. There is also a variant of wild
bootstrap based on the unrestricted OLS residuals $\hat{\varepsilon}$.
To save computation time, we omit the analysis of unrestricted wild
bootstrap, as earlier studies have repeatedly shown that restricted
wild bootstrap performs better. }
\item Generate bootstrap error terms $\varepsilon^{b}=\sqrt{\alpha^{\theta}}e^{r}\cdot v^{b}$,
where $\alpha^{\theta}$ is an adjustment based on type $\theta\in\{\text{HC1,HC2,HC3}\}$
as specified in Section 3. The factor $v^{b}$ is an $n\times1$ vector
of random weights independently drawn from a Rademacher distribution
\begin{equation}
v_{i}^{b}=\begin{cases}
-1 & \text{with probability }0.5\\
1 & \text{with probability }0.5
\end{cases}
\end{equation}
\item Form a bootstrap sample $y^{b}=\hat{y}^{r}+\varepsilon^{b}$ and re-estimate
the model to obtain the OLS estimator $\hat{\beta}_{k}^{b}$ and a
corresponding variance estimator $\hat{V}_{k}^{b,\eta}$ of type $\eta\in\{\text{HC1,HC2,HC3}\}$.
We then compute the corresponding t-statistic for the null hypothesis
$\beta_{k}=0$:
\begin{equation}
t_{k}^{b,\eta}=\frac{\hat{\beta}_{k}^{b}}{\sqrt{\hat{V}_{k}^{b,\eta}}}
\end{equation}

\item Repeat steps 2 and for 3 for $B$ bootstrap replications to construct
the bootstrap distribution of the test statistic.
\item Calculate the bootstrap p-value as the proportion of bootstrap statistics
$t_{k}^{b,\eta}$ that are as extreme as or more extreme than the
test statistic $t_{k}^{\eta}$ from the original regression sample
(also computed using a standard error of type $\eta$):
\begin{equation}
\text{p-value}=\frac{1}{B}\sum_{b=1}^{B}I(|t_{k}^{b,\eta}|\geq|t_{k}^{\eta}|),\label{eq:p_wild}
\end{equation}
where $I(\cdot)$ is the indicator function.
\end{enumerate}
We compare 9 different specifications $\tau\in\{\text{WB-11,WB-21},...,\text{WB-33\}}$
of wild bootstrap p-values, one for each combination $(\theta,\eta)\in\{\text{HC1,HC2,HC3}\}\times\{\text{HC1,HC2,HC3}\}$.

Similar to the previous Monte Carlo analysis we evaluate for each
original regression $M=10000$ Monte Carlo samples and we draw $B$
separate wild bootstrap samples for each Monte Carlo sample $m$.
Ideally, we would prefer to set $B$ to a large value in order to
precisely estimate bootstrap p-values for each Monte Carlo sample.
However, this approach presents a practical challenge. While wild
bootstrap p-values can be computed significantly faster than those
based on the paired bootstrap, the combination of a large $M$ and
$B$, together with over a thousand test situations, renders the computational
burden of the Monte Carlo study infeasible within acceptable time
frames given the limitations of our hardware.

We first derive a theoretical result that suggests that the rejection
rate $\pi_{\tau,s}^{0.05}$ for wild bootstrap methods can already
be well approximated with a smaller number of bootstrap repetitions.
Let $p_{\tau,s\vert B}(m)$ denote the bootstrap p-value for Monte
Carlo sample $m$ computed with $B$ bootstrap repetitions and let
\begin{equation}
\pi_{\tau,s\vert B}^{0.05}=\frac{1}{M}\sum_{m=1}^{M}I(p_{\tau,s\vert B}(m)\leq0.05).
\end{equation}
 We define the corresponding limit for infinitely many wild-bootstrap
replications as
\begin{equation}
P_{\tau,s}(m)=\text{plim}_{B\rightarrow\infty}p_{\tau,s}(m\vert B)
\end{equation}
and by furthermore taking the limit of infinitely many Monte Carlo
replications, we define
\begin{equation}
\Pi_{\tau,s}^{0.05}=\text{plim}_{M\rightarrow\infty}\frac{1}{M}\sum_{m=1}^{M}I(P_{\tau,s}(m)\leq0.05).
\end{equation}

\begin{prop}
Assume p-values $P_{\tau,s}$ are standard uniformly distributed and
$B$ is chosen such that $0.05\cdot(B+1)$ is an integer. Then
\begin{equation}
\text{plim}_{M\rightarrow\infty}\pi_{\tau,s\vert B}^{0.05}=\Pi_{\tau,s}^{0.05}
\end{equation}
and
\begin{equation}
\text{Var}(\pi_{\tau,s\vert B}^{0.05})\leq\frac{1}{4M}.
\end{equation}
\end{prop}

\begin{proof}
First, note that under the assumption that the p-values $P_{\tau,s}(m)$
are independently and identically standard uniformly distributed,
we have:

\begin{equation}
\Pi_{\tau,s}^{0.05}=0.05.
\end{equation}

Next, consider the bootstrap p-values $p_{\tau,s\vert B}(m)$ computed
with $B$ bootstrap replications. For each Monte Carlo sample $m$,
the bootstrap p-value $p_{\tau,s\vert B}(m)$ can be viewed as:

\begin{equation}
p_{\tau,s\vert B}(m)=\frac{K_{m}}{B},\quad\text{where }K_{m}\sim\text{Binomial}\left(B,P_{\tau,s}(m)\right).
\end{equation}

Since we assume $P_{\tau,s}(m)\sim U[0,1]$, the unconditional probability
mass function of $K_{m}$ is:

\begin{equation}
\mathbb{P}(K_{m}=k)=\int_{0}^{1}\binom{B}{k}p^{k}(1-p)^{B-k}dp=\binom{B}{k}\cdot\text{Beta}(k+1,B-k+1)
\end{equation}

Using the relationship between the Beta and Gamma functions and the
definition of the binomial coefficient:

\begin{equation}
\binom{B}{k}=\frac{B!}{k!(B-k)!},\quad\text{and}\quad\text{Beta}(a,b)=\frac{\Gamma(a)\Gamma(b)}{\Gamma(a+b)},
\end{equation}

we simplify:

\begin{equation}
\binom{B}{k}\cdot\text{Beta}(k+1,B-k+1)=\frac{B!}{k!(B-k)!}\cdot\frac{k!(B-k)!}{(B+1)!}=\frac{B!}{(B+1)!}=\frac{1}{B+1}.
\end{equation}

Thus, we find that:

\begin{equation}
\mathbb{P}(K_{m}=k)=\frac{1}{B+1},
\end{equation}

which means $K_{m}$ is uniformly distributed over $\{0,1,\ldots,B\}$.
Consequently, $p_{\tau,s\vert B}(m)$ is uniformly distributed over
the $B+1$ values $\left\{ 0,\frac{1}{B},\frac{2}{B},\ldots,1\right\} $.
Since we assume $(B+1)\cdot0.05$ is an integer, we know that $p_{\tau,s\vert B}(m)\leq0.05$
if and only if $p_{\tau,s\vert B}(m)$ takes one of the $(B+1)\cdot0.05$
values $\{0,\frac{1}{B},...,\frac{(B+1)\cdot0.05-1}{B}\}$. Therefore

\begin{equation}
\mathbb{P}\left(p_{\tau,s\vert B}(m)\leq0.05\right)=\frac{(B+1)\cdot0.05}{B+1}=0.05=\Pi_{\tau,s}^{0.05}.
\end{equation}
For the following steps let us rename the indicator variable as following:
\begin{equation}
Z_{m}=I(P_{\tau,s}(m)\leq0.05)
\end{equation}

It follows from the Law of Large Numbers that

\begin{equation}
\text{plim}_{M\rightarrow\infty}\pi_{\tau,s\vert B}^{0.05}=\text{plim}_{M\rightarrow\infty}\frac{1}{M}\sum_{m=1}^{M}Z_{m}=\mathbb{E}(Z_{m})=\Pi_{\tau,s}^{0.05}.
\end{equation}

Finally, we note that the variance of a Bernoulli distributed random
variable can never exceed $\frac{1}{4}$ and thus

\begin{equation}
\text{Var}\left(\pi_{\tau,s\vert B}^{0.05}\right)=\frac{1}{M}\text{Var}(Z_{m})\leq\frac{1}{4M}.
\end{equation}
\end{proof}
Although the wild bootstrap p-values $P_{\tau,s}(m)$ are not exactly
uniformly distributed, Proposition 1 suggests that the rejection rates
$\Pi_{\tau,s}^{0.05}$ can be quite accurately estimated by $\pi_{\tau,s\vert B}^{0.05}$
even with a moderate number of bootstrap replications $B$ if we draw
a large number of Monte Carlo samples.

The Monte Carlo results shown in Figure \ref{fig:just_boot} use $B=99$
bootstrap replications and $M=10000$ Monte Carlo samples. Since the
computations remain time-intensive, we have reduced the number of
test situations to \ncoefboot\  by limiting the tests to a maximum
of three coefficients per original regression.

\begin{figure}
\begin{centering}
\includegraphics[scale=0.9]{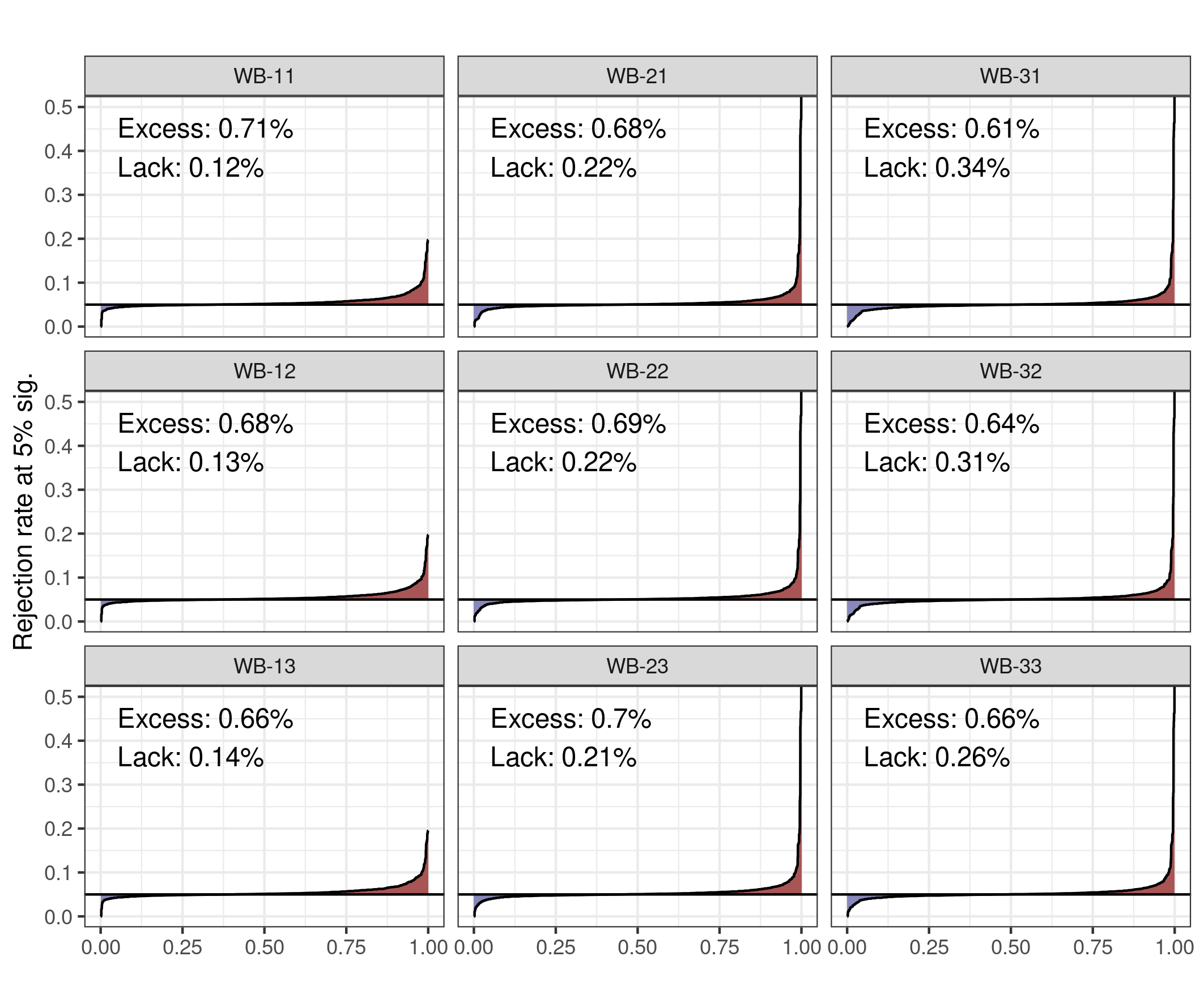}
\par\end{centering}
\caption{\label{fig:just_boot}Rejection rates for different wild bootstrap
specifications}
\end{figure}

Similar to the findings of MacKinnon et al. (2023) for the cluster-robust
wild bootstrap, the asymmetric specifications WB-31 and WB-13, which
utilize HC3 exclusively for the adjustment of the original OLS residuals
or solely for the computation of standard errors, respectively, tend
to exhibit superior performance. However, the performance differences
among the various wild bootstrap specifications are relatively minor.

Figure \ref{fig:boot} compares the rejection rates of the three bootstrap
specifications WB-11, WB-31, and WB-13 with the alternative specifications,
focusing on the \ncoefboot\  test situations for which bootstrap
standard errors are evaluated. While the wild bootstrap specifications
exhibit lower average excess in rejection rates than the HC1 and HC2
specifications, their average excess remains higher than that of HC3,
HC2-BM, JK-H, and HC2-PL.

\begin{figure}
\begin{centering}
\includegraphics[scale=0.9]{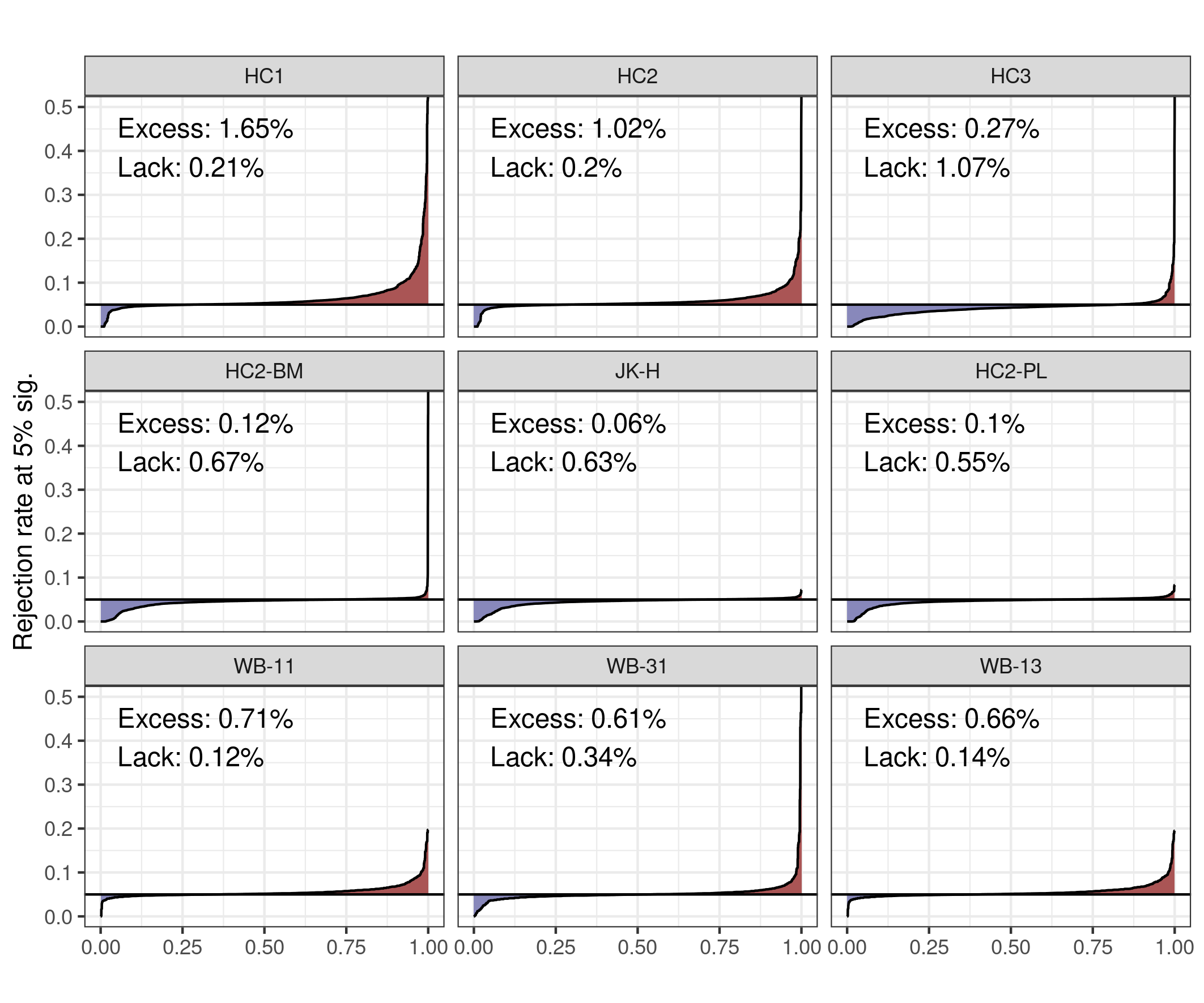}
\par\end{centering}
\caption{\label{fig:boot}Comparing rejection rates for wild bootstrap specifications
with non-bootstrap specifications}
\end{figure}

\pagebreak{}

\section*{Appendix E: Motivating HC1-PL and HC2-PL by a Satterthwaite Approximation}

\setcounter{table}{0}
\renewcommand{\thetable}{E\arabic{table}}
\setcounter{figure}{0}
\renewcommand{\thefigure}{E\arabic{figure}}

This appendix provides an alternative motivation of the proposed degree
of freedom adjustment of HC1-PL and HC2-PL based on $\tilde{n}_{k}$.
We start with the regression model
\begin{equation}
y=X\beta+\varepsilon
\end{equation}
with a deterministic regressor matrix $X$ and independently, normally
distributed, heteroskedastic errors

\begin{equation}
\varepsilon_{i}\sim N(0,\sigma_{i}^{2}).
\end{equation}

Let $\hat{V}_{k}^{\tau}$ be a variance estimator of the coefficient
$\hat{\beta}_{k}$ and consider the t-test for the null hypothesis
$\beta_{k}=0$ with test statistic
\begin{equation}
t_{k}=\frac{\hat{\beta_{k}}}{\sqrt{\hat{V}_{k}^{\tau}}}
\end{equation}

Following the approximation proposed by Satterthwaite (1946) $t_{k}$
follows approximately a t-distribution with $v_{k}$ degrees of freedom
satisfying (see e.g. Christensen, 2018 for a derivation):
\begin{equation}
\nu{}_{k}^{\tau}=\frac{2\text{Var}(\hat{\beta}_{k})^{2}}{\text{Var}\left(\hat{V}_{k}^{\tau}\right)}\label{eq:sw}
\end{equation}

We now show that the the partial-leverage adjusted sample size $\tilde{n}_{k}$
approximates $\nu{}_{k}^{HC0}$. Ding (2021) confirms that heteroskedasticity-robust
HC0 variances can be directly computed from the FWL representation
of the regression:
\begin{equation}
\tilde{y}_{k}=\beta_{k}\tilde{x}_{k}+\tilde{\varepsilon}_{k},
\end{equation}

yielding
\begin{equation}
\hat{V}_{k}^{HC0}=\frac{\sum_{i=1}^{n}\hat{\varepsilon}_{i}^{2}\tilde{x}_{k,i}^{2}}{\left(\sum_{i=1}^{n}\tilde{x}_{k,i}^{2}\right)^{2}}.
\end{equation}

Recall that the original OLS residuals $\hat{\varepsilon}_{i}$ are
the same as in the FWL specification. We first aim to find

\begin{equation}
\text{Var}(\hat{V}_{k}^{HC0})=\text{Var}\left(\frac{\sum_{i=1}^{n}\hat{\varepsilon}_{i}^{2}\tilde{x}_{i}^{2}}{\left(\sum_{i=1}^{n}\tilde{x}_{i}^{2}\right)^{2}}\right)=\frac{\text{Var}(\sum_{i=1}^{n}\hat{\varepsilon}_{i}^{2}\tilde{x}_{i}^{2})}{\left(\sum_{i=1}^{n}\tilde{x}_{i}^{2}\right)^{4}}.
\end{equation}

We approximate the variance of the sum in the numerator by the sum
of variances
\begin{equation}
\text{Var}(\sum_{i=1}^{n}\hat{\varepsilon}_{i}^{2}\tilde{x}_{k,i}^{2})\approx\sum_{i=1}^{n}\text{Var}\left(\hat{\varepsilon}_{i}^{2}\tilde{x}_{k,i}^{2}\right)=\sum_{i=1}^{n}\tilde{x}_{k,i}^{4}\text{Var}\left(\hat{\varepsilon}_{i}^{2}\right).
\end{equation}

If the formula were based on the independently distributed $\varepsilon_{i}$
instead of $\hat{\varepsilon}_{i}$, the approximation would be exact.
Yet, OLS residuals $\hat{\varepsilon}_{i}$ can in general be correlated
with each other. However, under the common assumptions for a consistent
OLS estimator, particularly strong exogeneity, $E(\varepsilon\mid X)=0$,
the OLS estimator $\hat{\beta}$ converges to $\beta$ as the sample
size increases and for $\hat{\beta}=\beta$ we have $\hat{\varepsilon}=\varepsilon$.

We now further approximate

\begin{equation}
\text{Var}\left(\hat{\varepsilon}_{i}^{2}\right)\approx\text{Var}\left(\varepsilon_{i}^{2}\right)=\mathbb{E}\left[\varepsilon_{i}^{4}\right]-\left(\mathbb{E}\left[\varepsilon_{i}{}^{2}\right]\right)^{2}=3\sigma_{i}^{4}-\sigma_{i}^{4}=2\sigma_{i}^{4}
\end{equation}

where we use the fact that normally distributed errors $\varepsilon_{i}$
satisfy

\begin{equation}
\ensuremath{\mathbb{E}\left[\varepsilon_{i}^{4}\right]=3\sigma_{i}^{4}}
\end{equation}

In a further simplification, we follow Bell and McCaffrey (2002) and
evaluate the resulting expressions for the case of homoskedasticity
\[
\sigma_{i}=\sigma\ \forall i=1,...,N
\]

We then find

\[
\text{Var}(\hat{V}_{k}^{HC0})\approx2\sigma^{4}\frac{\sum_{i=1}^{n}\tilde{x}_{k,i}^{4}}{\left(\sum_{i=1}^{n}\tilde{x}_{k,i}^{2}\right)^{4}}.
\]
The numerator of (\ref{eq:sw}) under homoskedasticity is given by

\[
2\text{Var}(\hat{\beta}_{k})^{2}=2\sigma^{4}\frac{\left(\sum_{i=1}^{n}\tilde{x}_{k,i}^{2}\right)^{2}}{\left(\sum_{i=1}^{n}\tilde{x}_{k,i}^{2}\right)^{4}}.
\]

We thus can approximate the degrees of freedom as
\[
\nu_{k}^{HC0}=\frac{2\text{Var}(\hat{\beta}_{k})^{2}}{\text{Var}\left(\hat{V}_{k}^{HC0}\right)}\approx\frac{\left(\sum_{i=1}^{n}\tilde{x}_{k,i}^{2}\right)^{2}}{\sum_{i=1}^{n}\tilde{x}_{k,i}^{4}}=\tilde{n}_{k}
\]

It is clear that this derivation involves several approximations,
and we would not propose our degree of freedom adjustment solely based
on this result, especially since we suggest using $\tilde{n}_{k}-1$
degrees of freedom instead of $\tilde{n}_{k}$. Nonetheless, the derivation
provides additional insight into why our proposal may be a reasonable
choice.

The main differences between this derivation and that of Bell and
McCaffrey (2002) are as follows: First, their computation of degrees
of freedom is based on an HC2 correction. Second, they do not use
the Frisch-Waugh-Lovell (FWL) representation as the starting point
for their approximation.
\end{document}